\documentclass[submission,copyright,creativecommons]{eptcs}
\usepackage{amsmath,amsthm,amsfonts,xspace,float,mathtools,adjustbox,enumerate,
centernot,graphicx}
\DeclarePairedDelimiter{\ceil}{\lceil}{\rceil}
\usepackage{lstautogobble}
\usepackage[dvipsnames,table]{xcolor}
\hypersetup{
    colorlinks,
    linkcolor={red!50!black},
    citecolor={blue!50!black},
    urlcolor={blue!80!black}
}
\usepackage[T1]{fontenc}
\usepackage{tikz}
\usetikzlibrary{positioning}

\newcommand*{\funcfont}{\fontfamily{lmss}\selectfont}
\newcommand*{\codefont}{\ttfamily\small}

\tikzstyle{thmbox} = [
  draw=blue,
  rectangle,
  align=center,
  text width=3cm,
  inner sep=0.3cm
]

\DeclareTextFontCommand{\funcfontify}{\funcfont}
\DeclareTextFontCommand{\codefontify}{\codefont}
\newcommand{\codify}[1]{\ensuremath{\mbox{\codefontify{#1}}}}

\lstdefinelanguage{acl2s}{
  keywords={defdata,definecd,definec,defun,property,create-map*,create-reduce*},
  comment=[l]{;;},
}

\newtheorem{definition}{Definition}
\newcounter{rlemmacounter}
\newcounter{clemmacounter}
\newcounter{blemmamcounter}
\newcounter{blemmaacounter}
\newcounter{blemmabcounter}

\newtheorem{lemma}{Lemma}
\newtheorem{theorem}{Theorem}
\newtheorem{rlemma}[rlemmacounter]{Lemma}

\newtheorem{clemma}[clemmacounter]{Lemma}

\newtheorem{blemmam}[blemmamcounter]{Lemma}

\newtheorem{blemmaa}[blemmaacounter]{Lemma}

\newtheorem{blemmab}[blemmabcounter]{Lemma}

\providecommand{\event}{ACL2 2023} 
\usepackage{iftex}
\ifpdf
  \usepackage{underscore}         
  \usepackage[T1]{fontenc}        
\else
  \usepackage{breakurl}           
\fi
\title{A Case Study in Analytic Protocol Analysis in ACL2}
\author{Max von Hippel\footnote{Authors are listed alphabetically by last name.}
\institute{Northeastern University\\Boston, Massachusetts}
\email{vonhippel.m@northeastern.edu}
\and
Panagiotis Manolios
\institute{Northeastern University\\Boston, Massachusetts}
\email{p.manolios@northeastern.edu}
\and
Kenneth L. McMillan
\institute{University of Texas at Austin\\Austin, Texas}
\email{kenmcm@cs.utexas.edu}
\and
Cristina Nita-Rotaru
\institute{Northeastern University\\Boston, Massachusetts}
\email{c.nitarot@northeastern.edu}
\and
Lenore Zuck
\institute{University of Illinois Chicago\\Chicago, Illinois}
\email{zuck@uic.edu}
}

\newcommand{\titlerunning}{Real Analysis Using ACL2}
\newcommand{\authorrunning}{von Hippel et. al.}

\newcommand{\sample}{\mbox{\sf S}\xspace}
\newcommand{\srtt}{\mbox{\sf srtt}\xspace}

\newcommand{\rttvar}{\mbox{\sf rttvar}\xspace}
\newcommand{\rto}{\mbox{\sf rto}\xspace}

\newcommand{\acls}{\textsc{ACL2s}\xspace}

\hypersetup{
  bookmarksnumbered,
  pdftitle    = {\titlerunning},
  pdfauthor   = {\authorrunning},
  pdfsubject  = {FM},               
}

\begin{document}
\maketitle

\begin{abstract}
When verifying computer systems we sometimes want to study their asymptotic behaviors,
i.e., how they behave in the long run.
In such cases, we need \emph{real analysis},
the area of mathematics that deals with limits and the foundations of calculus.
In a prior work, we used real analysis in 
ACL2s 
to study the asymptotic 
behavior of the RTO computation, commonly used in congestion control algorithms 
across the  Internet.
One key component in our RTO computation analysis was proving in 
ACL2s that for all $\alpha \in [0, 1)$, the limit as $n \to \infty$ of $\alpha^n$ is zero.
Whereas the most
obvious proof strategy involves the logarithm, whose codomain includes
irrationals, by default ACL2 only supports rationals, 
which forced us to take a non-standard approach.
In this paper, we explore different approaches to proving the above result
in ACL2(r) and ACL2s, from the perspective of a relatively
new user to each.  
We also contextualize the theorem by showing how it allowed us to prove important
asymptotic properties of the RTO computation.
Finally, we discuss tradeoffs between the various proof strategies and 
directions for future research.
\end{abstract}

\section{Introduction}\label{sec:intro}
In contrast to purely mathematically oriented theorem provers, 
ACL2 is designed specifically with the verification of computer systems in mind.
This focus manifests in a variety of places across the ACL2 software landscape,
such as the automated proofs of termination based on context-calling graphs
provided by
ACL2s~\cite{DillingerMVM07,acl2s11,manolios2006termination}, the
integration of QuickLisp,  
or the development of advanced string-solving 
capabilities~\cite{kumar2021mathematical}.  It also manifests in what ACL2 lacks:
e.g., ACL2 does not support irrationals such as 
$e, \pi,$ or $\sqrt{2}$, meaning it cannot be used
to reason about the reals in full generality.  Although this limitation is 
often immaterial, it shows up when we want
to study the asymptotic behaviors of computer systems, meaning how they behave
in the long run.  In such cases we require
\emph{real analysis}, the area of mathematics that deals with limits and the 
foundations of calculus.  However, in general real analysis proofs are
riddled with \emph{real} numbers; e.g., the logarithm is often used in 
these proofs, and it is often the case that the logarithm of a rational number
will be irrational.  If all we want is to prove a real analysis 
result, we can use ACL2(r)~\cite{gamboa2001nonstandard}, the variant of ACL2 
that supports real numbers.  Unfortunately, ACL2(r) proofs cannot be imported into a normal
ACL2 environment because the two systems have conflicting underlying theories, 
e.g., in ACL2 it is a theorem that $\forall x \,::\, x^2 \neq 2$, while in 
ACL2(r), $\sqrt{2}$ is a number.
So, what if we have formalized and studied a computer system in 
ACL2, and now we want to look at its asymptotic behaviors, without needing to
port our model to ACL2(r)?  

In this work we focus on one such scenario, drawn from our recent work~\cite{karn}
studying Karn's algorithm~\cite{karn1987improving} and the related 
Retransmission TimeOut (RTO) computation~\cite{rfc6298}.
This work can be viewed as a companion paper to our work in~\cite{karn},
with a special focus on the proof strategy for the RTO observations.
Next, we briefly summarize our work in~\cite{karn} as it provides context for this study.

\subsection{Motivating Example}

In order to understand the RTO computation and its motivation, 
we first need to understand the context in which it is used.
In the real world this context is an Internet connection between two computers,
such as might occur between a sender and receiver using TCP.
In the most general sense, the context consists of 
two endpoints communicating over a \emph{channel}.
The endpoints send and receive \emph{datagrams} partitioned into
\emph{packets} (that the sender transmits to the receiver) and \emph{acknowledgments}, or ACKs
(that the receiver transmits to the sender), and each datagram is uniquely identifiable
by its natural id and type (packet or ACK).
The channel is responsible for delivering messages transmitted from one endpoint to the other.
However, it might not do so reliably.
It cannot create new messages, but may reorder, drop, 
delay, or duplicate transmitted ones.\footnote{This model is less strict than the typical IP one where duplication is disallowed.}

We assume the sender does not transmit a packet $p > 1$ unless it previously transmitted all packets in
$[1, p-1]$, although it may transmit $p=1$ at any time.
We additionally assume that ACKs are cumulative in the sense that, 
the receiver does not transmit an ACK $a > 1$ unless it was previously delivered
packets $1, 2, \ldots, a-1$ but not $a$.
In the event that the receiver cannot cumulatively acknowledge anything, 
for instance, if it was delivered the packet 2 but never 1,
then it may transmit the trivial acknowledgment of 1 (which does not acknowledge any packets).
We also assume the receiver transmits an ACK whenever it is delivered a packet.
When the sender receives an acknowledgment~$a$, it considers~$a$ to be \emph{new} iff $a$ exceeds all the ACKs it received previously.  In other words,~$a$ is new iff it acknowledges at least one previously un-acknowledged packet.

In the real world, the channel may have limited bandwidth, and will start to lose datagrams
when its queues become full\footnote{Such lossy communication is commonly modeled using a so-called ``token bucket filter''.}.  This bandwidth is unknowable to either endpoint,
so the sender is forced to use ACKs to
assess the instantaneous state of the channel and react accordingly.  It does so in two ways.
First, the sender measures the round-trip time (RTT) using its local clock between when it first transmits a packet~$p$ and when it first receives any acknowledgment~$a > p$, as an indication of the pace at which the channel is delivering datagrams.  It can use this information to moderate its transmission rate.
Second, if no new ACKs arrive for some amount of time, the sender can assume the channel has been  overwhelmed and is dropping data.  In this case, it can slow its pace of transmission, and begin retransmitting unacknowledged data accordingly.  The time the sender will wait before timing out, slowing its pace, and retransmitting, is called the RTO and defined in RFC6298~\cite{rfc6298}.

The conjunction of these two mechanisms creates a problem:
suppose the sender retransmits an unacknowledged packet~$p$, then receives some new
ACK~$a > p$.  How does it know which transmission of~$p$ triggered the ACK?  
The estimated RTT will differ depending on the answer.
We illustrate this situation in Figure~\ref{fig:ambiguity}\footnote{Icons are from \url{https://openmoji.org/}}.
One solution, known as Karn's algorithm, is to only sample RTTs for packets that were transmitted precisely once, since ACKs for these packets are unambiguous~\cite{karn1987improving}.

\begin{figure}
\centering  
\includegraphics[width=0.7\textwidth]{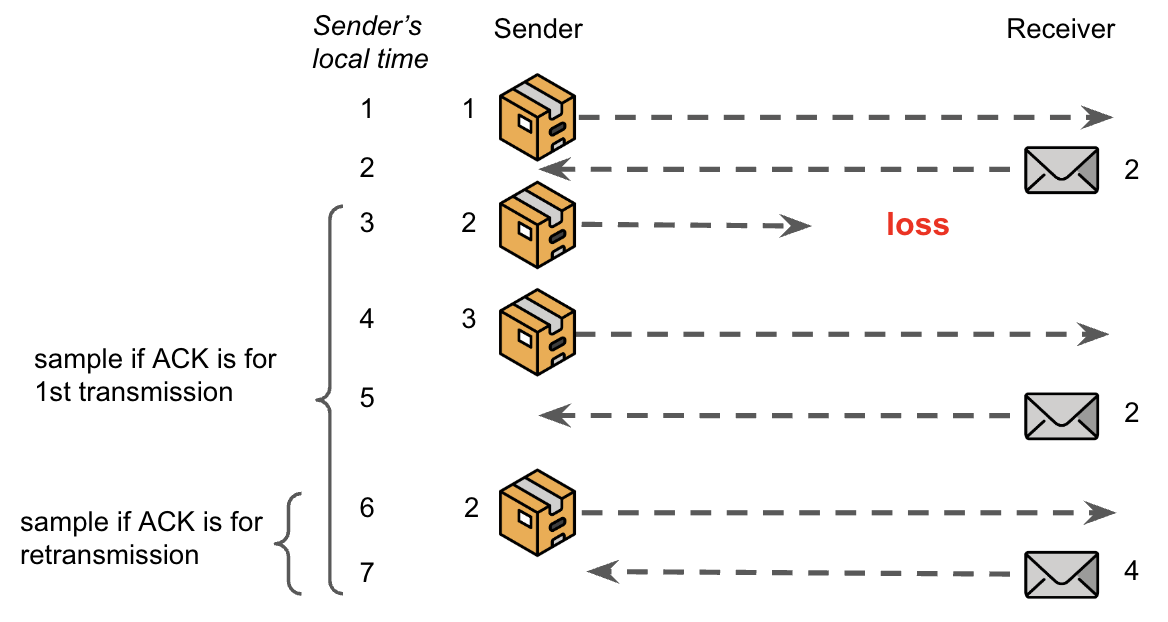}
\caption{Message sequence chart illustrating an ambiguous ACK, with the sender's local clock shown on the left. Sender's packets are illustrated as packets, while receiver's ACKs are shown as envelopes.  The first time the sender transmits 2 the packet is lost in-transit. Later, upon receiving a cumulative ACK of 2, the sender determines the receiver had not yet received the 2 packet and thus the packet might be lost in transit. It thus retransmits 2. Ultimately the receiver receives the retransmission and responds with a cumulative ACK of 4. When the sender receives this ACK it cannot determine which 2 packet delivery triggered the ACK transmission and thus, it does not know whether to measure an RTT of 7-3=4 or 7-6=1. Hence, the ACK is ambiguous, so any sampled RTT would be as well.}
\label{fig:ambiguity}
\end{figure}

We formally modeled Karn's algorithm in Ivy~\cite{padon2016ivy} with the network model outlined above.  We proved various inductive invariants about the algorithm, including that it samples a real RTT, that this RTT is in some sense pessimistic, and that when the receiver-to-sender channel path is FIFO, this RTT is for the packet whose id is equal to the previously highest-received acknowledgment.
Then, we formally modeled the RTO computation in ACL2s.
We chose ACL2s over Ivy because the computation is defined over real numbers, which we chose to model as rationals, and Ivy only supports integers.\footnote{In practice, the implementations we are aware of use integers, however, when Karn and Partridge wrote the algorithm down on paper, they did so using reals.}
And in particular, we chose ACL2s over ACL2 because we made frequent use of its features.
For example, we use the built-in counterexample generation to show that a variable referred to 
as the ``RTT variance'' is not actually a statistical variance;
and in one of our proofs, we use an automated proof of function termination to prove the existence 
of a particular value (namely, the value returned by the function in question).

The RTO computation is recursively defined over the RTT samples $\sample_1, \sample_2, \ldots$ output by Karn's algorithm 
and parameterized by three positive constants ($\alpha < 1$, $\beta < 1$, and~$G$)
as follows.
\begin{equation}
\begin{aligned}
\rto_i & = \srtt_i + \max(G, 4\cdot \rttvar_i) \\
\rttvar_i & = \begin{cases}
\sample_i / 2 & \text{ if } i = 1 \\
(1-\beta) \rttvar_{i-1} + \beta \lvert \srtt_{i-1} - \sample_i \rvert & \text{ if } i > 1
\end{cases}\\
\srtt_i & = \begin{cases} 
    \sample_i & \text{ if } i = 1 \\
    (1-\alpha)\srtt_{i-1} + \alpha \sample_i & \text{ if } i > 1 
    \end{cases} \\
\end{aligned}
\label{eqn:rto}
\end{equation}
Note, we use ``RTO'' when discussing the calculation generally, and ``\rto'' when discussing its actual implementation (given in the equation above).

We looked at what we called the \emph{steady-state} where for some 
rational \emph{center}~$c$ and \emph{radius}~$r$, the samples 
$\sample_i, \sample_{i+1}, \ldots, \sample_{i+n}$ all fall 
within the bounds $[c-r, c+r]$, and we proved the following.
\begin{enumerate}[I.]
    \item The $\srtt_{i+n}$ is bounded by the interval $[L, H]$ defined as follows.
    \begin{equation}
    \begin{aligned} 
    L & = (1-\alpha)^{n+1} \srtt_{i-1} + (1-(1-\alpha)^{n+1})(c-r) \\
    H & = (1-\alpha)^{n+1} \srtt_{i-1} + (1-(1-\alpha)^{n+1})(c+r)
    \end{aligned}
    \label{eqn:bound:srtt}
    \end{equation}
    \item For all $0 \leq m < n$, $\rttvar_{i+n}$ is upper-bounded by the following expression.
    \begin{equation}
    \begin{aligned}
    (1-\beta)^{n+1-m} \rttvar_{i+m-1} 
    + (1-(1-\beta)^{n+1-m}) \Delta_m & \text{ , where } \\
    \Delta_m = (1-\alpha)^{m+1}\srtt_{i-1} + 2r - (1-\alpha)^{m+1}(c+r)
    \end{aligned}
    \label{eqn:bound:rttvar}
    \end{equation}
\end{enumerate}

Then we reanalyzed the results in the asymptotic case.
In other words, we asked what these bounds converge to as the number~$n$ of 
consecutively bounded samples, as well as the cutoff $m < n$ for the bound~$\Delta_m$ defined above,
grow toward infinity.
By \emph{limit}, we are referring to the standard definition\footnote{Limit proofs that assume an $\epsilon > 0$ and then prove the existence of a corresponding $\delta > 0$ satisfying this definition are commonly reffered to as $\epsilon/\delta$-proofs.} from real analysis, 
which we give below using the 1D Euclidean metric $d(x, y) = \lvert x - y \rvert$.
\begin{definition}[Limit to $\infty$]
Let $(a_i)_{i=0}^{\infty} \in \mathbb{R}^\omega$ and $\ell \in \mathbb{R}$.
Then $\lim_{i \to \infty} a_i = \ell$ \emph{iff} the following holds:
\[\forall \epsilon > 0 \,::\, \exists \delta > 0 \,::\, \forall n > \delta \, :: \, \lvert a_n - \ell \rvert < \epsilon\]
\label{def:limit}
\end{definition}
We first proved the following theorem, which is the focus of this paper.
\begin{theorem}\(\forall \alpha \in [0, 1) \, :: \, \lim_{n \to \infty} \alpha^n = 0\)
\label{thm:a-n-0}
\end{theorem}
Theorem~\ref{thm:a-n-0} can be manually proven as follows.
\begin{proof}
Let $\epsilon > 0$ and $0 \leq \alpha < 1$ arbitrarily. 
If $\alpha=0$ the result is immediate; suppose $\alpha>0$.
Suppose $\epsilon < 1$, noting that if the theorem
holds for $\epsilon < 1$ then it holds for $\epsilon \geq 1$.
Let $\delta = \ln(\epsilon)/\ln(\alpha)$.
Note that $\ln(\epsilon)$ and $\ln(\alpha)$ are negative.
Let $n$ be some natural number and observe that $n \ln(\alpha) = \ln(\alpha^n)$.
Thus:
\[
\begin{aligned}
n > \delta & \iff n > \ln(\epsilon)/\ln(\alpha) & \text{ by definition of }\delta\\
           & \iff n \ln(\alpha) < \ln(\epsilon) & \text{ multiplying each side by }\ln(a)\\
           & \iff e^{n \ln(\alpha)} < e^{\ln(\epsilon)} & \text{ raising each side above }e\\
           & \iff e^{\ln(\alpha^n)} < \epsilon & \text{ because }e^{\ln(x)}=x\text{ for all }x\text{, and }n \ln(\alpha) = \ln(\alpha^n)\\
           & \iff \alpha^n < \epsilon & \text{ because }e^{\ln(x)}=x\text{ for all }x
\end{aligned}
\] 
\end{proof}
We then used this theorem to show that 
\(\lim_{n \to \infty} L = c-r\),
\(\lim_{n \to \infty} H = c+r\),
and the limit as $n$ and $m < n$ both grow toward infinity of the upper bound
in Eqn.~\ref{eqn:bound:rttvar} is precisely $2r$.

\paragraph{Outline.}
The rest of this paper is organized as follows.
We study Theorem~\ref{thm:a-n-0} in Sections~\ref{sec:real} and~\ref{sec:rational}.
Specifically, in Section~\ref{sec:real} we formalize the English-language proof given
above in ACL2(r).  But recall, we cannot use an ACL2(r) proof to study the RTO system 
we defined in ACl2s, without totally remodeling it, as the two proof systems are incompatible.
Thus in order to have our asymptotic proofs in the same model as our other preexisting 
proofs about the RTO system, we need
a rational proof.  We give two such proofs in Section~\ref{sec:rational}.
The first uses the ceiling function to define a $\delta$ directly.
The second begins by proving that $\lim_{n\to\infty} 1/2^n = 0$, then uses the binomial theorem
to construct a $\delta$ such that $n > \delta \implies \alpha^\delta < 1/2$.
For the second proof strategy, we show two different ways to prove $\lim_{n\to\infty} 1/2^n = 0$,
one of which is more automatic than the other.
With those proofs out of the way, we show how to derive the limits for the bounds on 
$\srtt$ and $\rttvar$ in Section~\ref{sec:context}, which concludes our analytic study of the RTO.
We discuss trade-offs between the various proofs and lessons learned in Section~\ref{sec:discussion} and conclude in Section~\ref{sec:conclusion}.

\section{Real Proof}\label{sec:real}
In this Section we overview the most obvious proof strategy for Theorem~\ref{thm:a-n-0} 
-- the one we gave in the introduction -- 
and its formalization in ACL2(r), the variant of ACL2 that supports
real numbers.  Because the rationals form a dense subset of the reals, 
this proof implies the desired result over the rationals as well.
However, since ACl2 and ACL2(r) are theoretically incompatible,
we cannot just import the proof into our preexisting ACL2s model to cohabitate
with our other theorems about the RTO calculation.

The theorem we aim to prove uses an existential quantifier, so we define it via \codify{defun-sk}.
Note that our theorem statement will be the same in Section~\ref{sec:rational}, except
that since we will be using ACL2s in those proofs, we will also have type declarations and guards there.
Notice how we can drop the absolute value signs from Definition~\ref{def:limit} because $\alpha$ is assumed to be positive, implying that $\alpha^n$ is also positive.
\begin{lstlisting}
(defun-sk lim-0 (a e n)
  (exists (d)
    (=> (^ (realp e) (< 0 e) (< d n)) (< (raise a n) e))))

(defthm lim-a^n->0
  (=> (^ (realp a) (< 0 a) (< a 1) (realp e) (< 0 e) (natp n))
      (lim-0 a e n)) :instructions ...) ;; proof will go here
\end{lstlisting}

The most important step in an $\epsilon/\delta$ proof is defining the $\delta$.
We do so by defining a witness function $d_a : \epsilon \to \delta$, so that the proof
obligation reduces to showing 
$\forall \epsilon > 0 \, :: \, \forall n > d_a(\epsilon ) \, :: \, \alpha^n < \epsilon$.
\begin{lstlisting}
(defun d (eps a) (/ (acl2-ln eps) (acl2-ln a)))
\end{lstlisting}

The remainder of the proof consists of two important steps.  First, we define a number of arithmetic lemmas which ACL2s proves automatically.  Second, because ACL2(r) lacks a generic real logarithm or exponent (having only the natural variants), we prove a translational lemma (\ref{lem:e-n-a-to-a-n}) saying that $e^{n \ln(\alpha)} = \alpha^n$.  Then we rephrase the proof from $\mathsection$\ref{sec:intro} in terms of~$e$, at which point it goes through easily.

\subsection{Arithmetic Lemmas}
We began by proving some basic arithmetic lemmas, which we needed for the more complicated proofs.
We proved that if $e^y < 1$ then $y < 0$, and that if $y \in (0, 1)$, then $\ln(y) < 0$.
Then we proved two facts about fractions of logarithms.
First, if $\alpha \in (0, 1)$ then \((\ln(\epsilon)/\ln(\alpha)) \ln(\alpha) = \ln(\epsilon)\).
Second, if $\epsilon$ and $\alpha$ both fall within $(0, 1)$
and $n > \ln(\epsilon)/\ln(\alpha)$, then \(n \ln(\alpha) < \ln(\epsilon)\),
and thus, 
since the natural exponent is monotonic (i.e., $x < y \implies e^x < e^y$),
it follows that \(e^{n \ln(\alpha)} < e^{\ln(\epsilon)} = \epsilon\).
Finally, combining these results gave us that if $\alpha$ and $n \in \mathbb{R}_+$ then
$\ln(\alpha^n) = \ln(e^{n \ln(\alpha)}) = n \ln(\alpha)$.

\subsection{Translational Lemma}
Our arithmetic lemmas allow us to prove the following translational result.
\begin{rlemma}
For all positive reals $\alpha$ and $n$, $e^{n \ln(\alpha)} = \alpha^n$.
\label{lem:e-n-a-to-a-n}
\end{rlemma}

\subsection{Proof of Theorem~\ref{thm:a-n-0}}
Next we prove the desired result without loss of generality, and without explicitly using quantifiers.
Our proof uses Lemma~\ref{lem:e-n-a-to-a-n} as a hint.
\begin{rlemma}
Let $\alpha$ and $\epsilon$ be reals in $(0, 1)$ and let $n > d_\alpha(\epsilon)$ be a natural.
Then $\alpha^n < \epsilon$.
\label{lem:a-n-0-helper}
\end{rlemma}
Notice how this immediately gives us our desired result, because if $\epsilon < 1$ then we can just
use Lemma~\ref{lem:a-n-0-helper} directly, and if $\epsilon \geq 1$, we can use it with a new ``epsilon'' $< \epsilon$.  And this is precisely the strategy we take in the \codify{instructions} to our proof of Theorem~\ref{thm:a-n-0}, which go as follows:
\begin{enumerate}[i.]
    \item Promote all variables then case-split on $\epsilon < 1$.
    \item Case 1: $\epsilon < 1$.  Use Lemma~\ref{lem:a-n-0-helper}.  Then instantiate \codify{lim-0-suff} with $\delta = d_\alpha(\epsilon)$ and prove.
    \item Case 2: $\epsilon > 1$.  Use Lemma~\ref{lem:a-n-0-helper} with a new ``epsilon'' value of $\epsilon' = 1/2$ and claim that all its preconditions are satisfied.  Further claim that if $d_\alpha(1/2) < n$ then $\alpha^n < 1/2 < \epsilon$.  Then instantiate \codify{lim-0-suff} with $\delta = d_\alpha(1/2)$, promote, and prove.
\end{enumerate}

On the one hand, our proof is straightforward in the sense that it mostly follows the English-language
proof we outlined in the introduction.
On the other hand, we can easily see some places where more machinery and theorems in the nonstandard arithmetic
library would drastically simplify things.
The biggest ommision is that the nonstandard arithmetic library only supports natural exponent and logarithm, which forced us to use the translational lemma.  Most interestingly, this is not the shortest proof in this paper!  There is actually a more concise\footnote{(as measured by lines of code and number of imported books)}, rational proof which we outline in the next Section.

\section{Rational Proofs}\label{sec:rational}
In this Section, we reprove Theorem~\ref{thm:a-n-0} in ACL2s, 
using only rationals.
We present two proofs.  
The first is the one we used in our motivating work~\cite{karn},
where we explicitly construct the~$\delta$ using the ceiling function.
To do so, we have to prove various properties of that function.
The second proof is our most concise, and proceeds in two steps.
First we show that $\lim_{n \to \infty} 1/2^n = 0$.
Then using the binomial theorem,
we show how, for any $0 < \alpha < 1$, we can construct an $n$ such that $\alpha^n < 1/2$.
The result follows.
For convenience, we refer to the first proof 
(given in $\mathsection$\ref{subsec:ceil}) as the \emph{ceiling proof}
and the second 
($\mathsection$\ref{subsec:binom}) as the \emph{binomial proof}.
We recap in $\mathsection$\ref{subsec:closing}.

\subsection{Ceiling Proof}\label{subsec:ceil}
In prose, the ceiling proof of Theorem~\ref{thm:a-n-0} proceeds as follows.
\begin{proof}
Let $0 \leq \alpha < 1$ and $\epsilon > 0$, arbitrarily.
Let $k = \ceil{a/(1-a)}$ and observe that $a \leq k/(k+1)$.
Let $f(n) = k\alpha^k/n$.  
As an intermediary lemma, we claim that for all $n \geq k$, $\alpha^n \leq f(n)$.
\begin{description}
    \item \emph{Base Case}: $n = k$ thus $f(n) = \alpha^k \geq \alpha^n$ and we are done.  
    \item \emph{Inductive Step}: By inductive hypothesis, we have
    \begin{equation}
    a^n \leq k \alpha^k / n
    \label{eqn:ind}
    \end{equation}
    and $k \leq n$.
    This gives us $k/(k+1) \leq n/(n+1)$ and thus:
    \begin{equation}
    \alpha \leq n/(n+1)
    \label{eqn:a-ineq}
    \end{equation}
    Multiplying Eqn.~\ref{eqn:ind} through by $\alpha$, we get
    \(
    \alpha^{n+1} \leq k \alpha^{k+1} / n
    \).
    Combining this with Eqn.~\ref{eqn:a-ineq}:
    \begin{equation}
    \alpha^{n+1} \leq ( k \alpha^k / n ) \frac{n}{n+1} = k \alpha^k / (n+1)
    \end{equation}
    and we are done.
\end{description}
Hence induction: $\forall n \geq k$, $\alpha^n \leq f(n)$.  Now, 
let $\delta = \ceil{k\alpha^k/\epsilon}$.  It follows that 
$\forall n \geq \delta$, $f(n) \leq \epsilon$, and thus by the above result,
$\alpha^n \leq \epsilon$.  We get $\alpha^n < \epsilon$ by repeating this process
for $\epsilon / 2$, and we are done.
\end{proof}
Although the proof is relatively straightforward on paper,
as we will see, it is much more challenging in ACL2s.
The primary issue is that ACL2/ACL2s does not by default know very much about the ceiling function,
so we will be forced to prove many obvious lemmas before making the essential
argument.
Since we are in ACL2s now, we first restate Theorem~\ref{thm:a-n-0} with types.
\begin{lstlisting}
(defun-sk lim-0 (a e n)
  (declare (xargs :guard (and (posratp a) (< a 1) (posratp e) (natp n))
          :verify-guards t))
  (exists (d) (and (natp d) (implies (< d n) (< (expt a n) e)))))

(property lim-a^n->0 (a e :pos-rational n :nat)
  :hyps (< a 1)
  (lim-0 a e n) :instructions ...) ;; proof will go here
\end{lstlisting}
The rest of the subsection is organized in follows.
We prove arithmetic lemmas in 
$\mathsection$\ref{subsubsec:ceil-arith}.
We use these lemmas to prove the ``intermediary lemma'' 
in $\mathsection$\ref{subsubsec:ceil-int},
which we use to prove 
prove Thm.~\ref{thm:a-n-0}
in $\mathsection$\ref{subsubsec:ceil-a-n-0}.

\subsubsection{Arithmetic Lemmas}\label{subsubsec:ceil-arith}
In order to fill out the \codify{instructions}, we first need some lemmas,
primarily about the ceiling function.
\begin{clemma}
For all $x, y \in \mathbb{Q}_+$, if $\ceil{x} < \ceil{y}$ then (i) $x \leq \ceil{x}$ and (ii) $\ceil{x} < y$.\label{lem:cx-le-cy-then-x-leq-cx-and-cx-le-y}
\end{clemma}
For the next Lemma we write a manual proof, adapted from~\cite{233684}, and utilizing Lemma~\ref{lem:cx-le-cy-then-x-leq-cx-and-cx-le-y}, which we give immediately below.
\begin{clemma}
Let $m, n \in \mathbb{N}_+$ and $x \in \mathbb{Q}_+$.
Then $\ceil{x/mn} = \ceil{\ceil{x/m}/n}$.\label{lem:cxmn-ceil-ceilxm-n}
\end{clemma}
\begin{proof}
Observe:
\begin{equation}
\ceil{x/m} - 1 < x/m \leq \ceil{x/m}
\label{eqn:cxm-1-le-xm-leq-cxm}
\end{equation}
Dividing Eqn.~\ref{eqn:cxm-1-le-xm-leq-cxm} by $n$, we get 
\(
(\ceil{x/m} - 1)/n < x/mn \leq \ceil{x/m}/n
\).  We also observe that
\(
\ceil{x/mn} \leq \ceil{\ceil{x/m}/n}
\).
Thus by Lemma~\ref{lem:cx-le-cy-then-x-leq-cx-and-cx-le-y}:
\begin{equation}
x/mn \leq \ceil{x/mn} < \ceil{x/m}/n
\label{eqn:x-mn-leq-cx-mn-le-cxm-n}
\end{equation}
Suppose (for a contradiction) that $\ceil{x/mn} < \ceil{\ceil{x/m}/n}$.  
Multiplying Eqn.~\ref{eqn:x-mn-leq-cx-mn-le-cxm-n} by $n$, we get
\(
x/m \leq n \ceil{x/mn} < \ceil{x/m}
\),
which is sufficient information for ACL2s to automatically find the contradiction:
\begin{equation}
\ceil{x/m} \leq n \ceil{x/mn} < \ceil{x/m}
\end{equation}
sufficing to show that $\ceil{x/mn} \centernot{<} \ceil{\ceil{x/m}/n}$.
But then $\ceil{x/mn} = \ceil{\ceil{x/m}/n}$ and we are done.
\end{proof}
Next we observe that for all $x, y, z \in \mathbb{Q}_+$, if $y \leq z$ then $y/x \leq z/x$ and moreover, $yx \leq zx$.  The following property of the ceiling function automatically follows. 
Next we make an important observation about the ceiling function.
\begin{clemma}
Let $\alpha \in \mathbb{Q}$ such that $0 < \alpha < 1$.
Let $k = \ceil{\alpha/(1-\alpha)}$.  Then $\alpha \leq k/(1+k)$.
\end{clemma}
\begin{proof}
First note that $\alpha/(1-\alpha) \leq k$.
Observe that for all $x, y, z \in \mathbb{Q}_+$, if $y \leq z$,
then $yx \leq zx$.  It follows that 
\(\alpha = (\alpha/(1-\alpha))(1-\alpha) \leq k(1-\alpha)\).
Next observe that $\alpha \leq k(1-\alpha) = k - k\alpha$.
Adding $k\alpha$ to each side, we get $\alpha + k\alpha = \alpha(1 + k) \leq k$.
Again consider $x, y, z \in \mathbb{Q}_+$ such that $y \leq z$,
but this time, observe that $yx \leq zx$.
Thus, 
\(\alpha(1+k)/(1+k) = \alpha \leq k/(1+k)\), and we are done.
\end{proof}
Next we prove two lemmas about fractions.
\begin{clemma}
For all $k \in \mathbb{N}_+$ and $\alpha \in \mathbb{Q}_+$,
$k \alpha^k / k = \alpha^k$. 
\label{lem:ka-k-eq-a-k}
\end{clemma}
\begin{clemma}
For all $k \leq n \in \mathbb{N}, k/(1+k) \leq n/(1+n)$.
\label{lem:div-lem-ind-step}
\end{clemma}
Finally, we make some obvious arithmetic observations, leading to the following result. 
\begin{clemma}
For all $x, y \in \mathbb{Q}_+$, we have $x/\ceil{x/y} \leq y.$
\label{lem:inequality-over-ceil-frac}
\end{clemma}
\begin{proof}
Note $x/y \leq \ceil{x/y}$,
thus $1/\ceil{x/y} \leq 1/(x/y)$.
Multiply both sides by $x$, and we are done.
\end{proof}

\subsubsection{Inductive Proof of Intermediary Lemma}\label{subsubsec:ceil-int}
With these arithmetic lemmas completed we can move on to the actual proof.
For convenience, we will define a parameterized function 
$f_\alpha : \mathbb{N}_+ \to \mathbb{Q}_+$ 
such that $f_\alpha(n) = k \alpha^k / n$ for
$k = \ceil{\alpha/(1-\alpha)}$.
As an intermediary lemma, we claim that for all $n \geq k$, $a^n \leq f(n)$.
Assuming the lemma holds, we can 
let $\delta = \ceil{ka^k/\epsilon}$,
and we immediately get that for all 
$n \geq \delta$, $a^n \leq f_\alpha(n) \leq \epsilon$.
Theorem~\ref{thm:a-n-0} immediately follows.
This sub-subsection is spent proving the intermediary lemma.

\begin{clemma}[Base Case]
Let $\alpha \in \mathbb{Q}_+$ and let $k = \ceil{\alpha/(1-\alpha)}$.
Then $f_\alpha(k) = a^k$.
\label{lem:base-case-ken}
\end{clemma}
\begin{proof}
Follows directly from Lemma~\ref{lem:ka-k-eq-a-k}.
\end{proof}

Before the inductive step, we need one more helper lemma,
the proof of which follows from our prior arithmetic observations.
\begin{clemma}
For all $n, k \in \mathbb{N}_+$ and $\alpha \in \mathbb{Q}_+$,
if $\alpha^{n+1} \leq \alpha k \alpha^k / n$, then
$\alpha^{n+1} \leq k a^k / (1 + n)$.
\label{lem:incr-denom-lemma}
\end{clemma}

\begin{clemma}[Inductive Step]
Let $\alpha \in \mathbb{Q}_+$ and $n \in \mathbb{N}$.  Suppose that $\alpha < 1$,
$k = \ceil{\alpha/(1-\alpha)} \leq n$, and $\alpha^n \leq f_\alpha(n)$.
Then $\alpha^{1+n} \leq f_\alpha(1+n)$.\label{lem:inductive-step-kens-proof}
\end{clemma}
\begin{proof}
By Lemma~\ref{lem:base-case-ken}, $f_\alpha(k) = a^k$.
By Lemma~\ref{lem:div-lem-ind-step},
\(\alpha \leq k/(1+k) \leq n/(1+n)\).
Then by our arithmetic observations,
$\alpha^{n+1} \leq f_\alpha(k) \alpha$, and thus,
$\alpha^{n+1} \leq k a^k / (1 + n) = f_\alpha(1+n)$.
\end{proof}

Although at this point we've laid out our inductive argument, we still need to
implement it in ACL2s.  To begin with, this means defining an inductive scheme.
\begin{lstlisting}
(definec ikn (a :pos-rational n :nat) :nat
  :ic (< a 1)
  (if (> (ceiling (/ a (- 1 a)) 1) n) 0 (1+ (ikn a (- n 1)))))
\end{lstlisting}
We use this scheme in the proof of the next lemma.
\begin{clemma}[Intermediary Lemma]
Let $\alpha < 1$ be in $\mathbb{Q}_+$ and $n \geq \ceil{\alpha/(1-\alpha)}$ in $\mathbb{N}$.
Then $\alpha^n \leq f_\alpha(n)$.
\label{lem:n-geq-k-implies-a-to-n-leq-fn}
\end{clemma}
\begin{proof}
Induct on \codify{ikn}. Use Lemma~\ref{lem:base-case-ken} for the base case and Lemma~\ref{lem:inductive-step-kens-proof} for the inductive step.
\end{proof}

\subsubsection{Proof of Theorem~\ref{thm:a-n-0}}\label{subsubsec:ceil-a-n-0}
Our proof strategy is as follows.
First, we introduce a function $\delta_\alpha : \mathbb{R}_+ \to \mathbb{N}_+$
defined by $\epsilon \mapsto \max\{ k, d \}$, where $k = \ceil{\alpha/(1-\alpha)}$
as before, and $d = \ceil{k\alpha^k/\epsilon}$.
Then we prove three lemmas about this function (given immediately below) which together 
suffice to imply Theorem~\ref{thm:a-n-0}.  We use $k$ and $d$ as defined above.
\begin{clemma}
Let $\alpha < 1$ and $\epsilon$ be in $\mathbb{Q}_+$ and $n \in \mathbb{N}$.
Suppose $\delta_\alpha(\epsilon) \leq n$.
Then $k \leq n$.
\label{lem:n-geq-d-implies-n-geq-k}
\end{clemma}
\begin{proof}
By definition of $\delta_\alpha$, we have
$\max\{ k, d \} \leq n$.  Since $k \leq \max\{ k, d \}$, we are done.
\end{proof}
\begin{clemma}
Let $\alpha < 1$ and $\epsilon$ be in $\mathbb{Q}_+$ and $n \in \mathbb{N}$.
Suppose $\delta_\alpha(\epsilon) \leq n$.
Then $\alpha^n \leq f_\alpha(n)$.
\label{lem:a-n-to-fan}
\end{clemma}
\begin{proof}
Follows automatically from Lemmas~\ref{lem:n-geq-k-implies-a-to-n-leq-fn} and~\ref{lem:n-geq-d-implies-n-geq-k} with the definitions of $f_\alpha$ and $\delta_\alpha$.
\end{proof}
\begin{clemma}
Let $\alpha < 1$ and $\epsilon$ be in $\mathbb{Q}_+$ and $n \in \mathbb{N}$.
Suppose $\delta_\alpha(\epsilon) \leq n$.
Then $f_\alpha(n) \leq \epsilon$.
\label{lem:fan-to-0}
\end{clemma}
\begin{proof}
By the definition of $\delta_\alpha$, we have
\(\max\{ k, d \} \leq n\).
Our prior arithmetic observations give us that $k \alpha^k / n \leq k \alpha^k / d$.
Since $d = \ceil{k\alpha^k/\epsilon}$, 
by Lemma~\ref{lem:inequality-over-ceil-frac},
clearly $k \alpha^k / d \leq \epsilon$.
The result follows.
\end{proof}
Armed with these Lemmas, we can prove a ``helper lemma'' like we did previously.
But in this case, we use $\leq$ instead of $<$ because of the way we structured our
argument based on the intermediary result.
\begin{clemma}
For all $\alpha < 1$ and $\epsilon$ in $\mathbb{Q}_+$ and $n \in \mathbb{N}_+$,
if $\delta_\alpha(\epsilon) \leq n$ then 
$\alpha^n \leq \epsilon$.
\end{clemma}
\begin{proof}
Follows from Lemmas~\ref{lem:a-n-to-fan} and~\ref{lem:fan-to-0} after observing that
all their preconditions are met.
\end{proof}
Finally, we can provide the \codify{instructions} to prove Theorem~\ref{thm:a-n-0}.
Note how we divide $\epsilon$ by 2 in order to transform $\leq$ into $<$, to fit 
Definition~\ref{def:limit}.\footnote{N.b., this trick suffices to show that the alternative definition with $\leq$ is equivalent.}
\begin{lstlisting}
((:use (:instance lim-0-suff (d (delta a (/ e 2)))))
 (:use (:instance a^n->0 (a a) (e (/ e 2)) (n n)))
 :pro :prove)
\end{lstlisting}

\subsection{Binomial Proof}\label{subsec:binom}
In this subsection, we propose an alternative proof.  
The strategy can be split into two steps.
First, we prove that 
$0 \leq \alpha \leq 1/2 \implies \lim_{n \to \infty} \alpha^n = 0$.  
Second, we prove that
for all $\alpha \in [0, 1)$, there exists some $\delta \in \mathbb{N}$ such that
$n > \delta \implies \alpha^n \leq 1/2$.  
We rely on the binomial theorem to find this $\delta$, 
hence the name of the proof.
These results suffice to prove
Theorem~\ref{thm:a-n-0}.

We found two ways to approach the first step.
The first way was to attack the problem directly, with an $\epsilon/\delta$
proof.  The second was to leverage the termination
analysis in ACL2s to find a $\delta$ semi-automatically.
We cover the first approach in 
$\mathsection$\ref{subsubsec:eps-del-half-0}
and the second in 
$\mathsection$\ref{subsubsec:automatic-half-0}.
Then in
$\mathsection$\ref{subsubsec:bin-part-2}
we show how, given either approach, we can prove Thm.~\ref{thm:a-n-0} by completing the ``second step'' described above.

\subsubsection{Manual Proof of $0 \leq \alpha \leq 1/2 \implies \lim_{n \to \infty} \alpha^n = 0$.}
\label{subsubsec:eps-del-half-0}
We begin by importing the proof-by-arithmetic book.
\begin{lstlisting}
(include-book "make-event/proof-by-arith" :dir :system)
\end{lstlisting}
We then prove the a sequence of simple arithmetic facts, using some combination of the 
\codify{linear}, 
\\\codify{match-free}, and \codify{all} rule classes.
First, like we did in the prior section, we show that the exponent is monotonic.
Second, we show that for all $n \in \mathbb{N}$, $n < 2^n$, and thus, if $n$ is positive, then $1/2^n < 1/n$.
Then we introduce an arithmetic trick by which we can extract a number smaller than $\epsilon$,
namely, if $\epsilon = x/y$ is a positive rational, then $1/y < \epsilon$.
Combining these results yields the following two lemmas.
\begin{blemmam}
For all $\alpha \leq 1/2$ in $\mathbb{Q}_+$, and for all $d \in \mathbb{N}_+$,
$\alpha^d \leq 1/2^d$.
\label{lem:alpha-leq-half-impl-alpha-d-leq-half-d}
\end{blemmam}
\begin{blemmam}
For all $\alpha, \epsilon = x/y \in \mathbb{Q}_+$, where $x, y \in \mathbb{N}_+$,
if $\alpha \leq 1/2$,
then $\alpha^y \leq \epsilon$.
\label{lem:a-leq-half-impl-a-denom-eps-leq-eps}
\end{blemmam}
At this point, the desired result follows directly from 
Lemma~\ref{lem:a-leq-half-impl-a-denom-eps-leq-eps}.

\subsubsection{Semi-Automatic Proof of $0 \leq \alpha \leq 1/2 \implies \lim_{n \to \infty} \alpha^n = 0$.}
\label{subsubsec:automatic-half-0}
In the semi-automatic proof, we begin by defining two functions.
The first function, $\mu : \mathbb{N} \times \mathbb{N} \to \mathbb{N}$, 
is defined by $(b, q) \mapsto b \textit{ if } q < 2^b \textit{ else } \mu(b + 1, q)$,
and can be viewed as a ``helper function'' to the second,
$d : \mathbb{Q}_+ \to \mathbb{N}$,
which is defined by $\epsilon \mapsto \mu(0, \text{denominator}(\epsilon))$.
When we define these two functions, ACL2s automatically proves that they terminate,
meaning that for all possible inputs of $\epsilon \in \mathbb{Q}_+$,
$\mu(0, \epsilon)$ terminates.
We use this with an inverse argument to show $1/2^{d(\epsilon)} \leq \epsilon$,
which we then manipulate to get the desired result.
Because the argument and manipulation require writing down
additional lemmas, we call this proof \emph{semi}-automatic.
Next, we prove three lemmas about these functions.
The first two go through automatically, whereas the third requires a manual proof.
\begin{blemmaa}
For all $b, q \in \mathbb{N}$, $q < 2^{\mu(b, q)}$.
\label{lem:uplog-ups}
\end{blemmaa}
\begin{blemmaa}
For all $q \in \mathbb{N}_+$, $1/2^{\mu(0, q)} < 1/q$.
\label{lem:belowlow}
\end{blemmaa}
\begin{blemmaa}
For all $\epsilon > 0$, $1/2^{d(\epsilon)} < \epsilon$.
\label{belowlog-to-zero}
\end{blemmaa}
\begin{proof}First observe that $1/\text{denominator}(\epsilon) \leq \epsilon$.
Then observe that $1/2^{\text{denominator}(\epsilon)} < 1/\text{denominator}(\epsilon)$.
The rest follows automatically.
\end{proof}
Next, we establish the monotonicity of the exponent, both strictly ($<$) and otherwise ($\leq$).
This allows us to prove the following.
\begin{blemmaa}For all $k \leq n \in \mathbb{N}_+$ and $\alpha < 1$ in $\mathbb{Q}_+$, $\alpha^n \leq \alpha^k$.\end{blemmaa}
\begin{blemmaa}
For all $\epsilon > 0$ in $\mathbb{Q}_+$ and $n > d(\epsilon)$ in $\mathbb{N}$,
$1/2^n < \epsilon$.
\label{lem:half-to-n-to-0}
\end{blemmaa}
\begin{proof}
Follows from the (non-strict) monotonicity of the exponent,
Lemma~\ref{belowlog-to-zero}, and the observation that for all $n \in \mathbb{N}$, $1/2^n = (1/2)^n$.
\end{proof}
We prove the next lemma in its given form, and rewritten using \emph{numerator} and \emph{denominator}.
\begin{blemmaa}
For all $p, q \in \mathbb{N}_+$, if $p/q < 1$ then $p < q$.
\label{lem:p-over-q-le-1-implies-p-le-q}
\end{blemmaa}
\begin{blemmaa}
For all $x \leq y$ and $\alpha$ in $\mathbb{N}_+$, $\alpha/y \leq \alpha/x$.
\label{lem:inv-monotone-div}
\end{blemmaa}
Finally, we get the desired result.
\begin{blemmaa}
For all $\alpha \leq 1/2$ and $\epsilon$ in $\mathbb{Q}_+$,
and for all $n \geq d(\epsilon)$ in $\mathbb{N}$,
$\alpha^n < \epsilon$.
\label{lem:a-le-half-to-zero}
\end{blemmaa}
\begin{proof}
Follows from Lemmas~\ref{lem:half-to-n-to-0} and~\ref{lem:inv-monotone-div}.
\end{proof}
Having shown two ways to derive the first step of our outlined proof strategy,
we now move on to the second step, where we invoke the binomial theorem.

\subsubsection{Remainder of Binomial Proof}
\label{subsubsec:bin-part-2}
The remainder of the proof begins with the following lemma.  The lemma is the same regardless of which strategy we take for part 1 (i.e., manual, or semi-automatic), however, its proof differs slightly with each choice.  Thus for brevity, we only include the proof assuming we took the semi-automatic approach, since it is the most recent in the text and thus the easiest to compare to here.  The alternative version given the manual approach is very nearly identical.
\begin{blemmab}
For all $\alpha = x/y \in \mathbb{Q}_+$, where $x, y \in \mathbb{N}_+$, $\alpha \leq x/(1+x)$.
\label{lem:x-over-y-leq-x-over-1-plus-x}
\end{blemmab}
\begin{proof}
Follows from Lemmas~\ref{lem:p-over-q-le-1-implies-p-le-q} and~\ref{lem:inv-monotone-div}.
\end{proof}
Next we import the binomial theorem into our proof.
\begin{lstlisting}
(include-book "arithmetic/binomial" :dir :system)
\end{lstlisting}
Note that the \codify{binomial} book was written in ACL2.  Since we are in ACL2s,
we have an additional obligation to check types.  So, we prove four convenient
lemmas about the types involved in the binomial expansion as defined in that book.
\begin{blemmab}
The codomain of the \codify{choose} function is a subset of $\mathbb{Z}$.
\end{blemmab}
\begin{blemmab}
The integer-exponent of an integer is an integer.
\end{blemmab}
\begin{blemmab}
The binomial \codify{expansion} (defined in the \codify{binomial} book) is a list of naturals.
\end{blemmab}
\begin{blemmab}
For any list of naturals, its summation under the \codify{sumlist} function is a natural.
\end{blemmab}
Now we get to the crux of the argument.  Essentially, we will show that if $\alpha < 1$ is a rational with numerator $p$ and denominator $q$, then $\alpha = p/q \leq p/(p+1)$ and thus $\alpha^p \leq p^p/(p+1)^p$.  By the binomial theorem, $(p+1)^p \geq 2p^p$, thus $\alpha^p \leq 1/2$, allowing us to reduce to the argument from step 1.  Formally speaking, we accomplish this through the following sequence of lemmas.
\begin{blemmab}
For all $n \in \mathbb{N}_+$, $2n^n \leq (1+n)^n$.
\label{lem:2-n-n-leq-1-plus-n-n}
\end{blemmab}  
\begin{proof}
Follows from the binomial theorem because $(1+n)^n \leq 1 + \ldots + n n^{n-1} + n^n$.
\end{proof}

Now we prove a ``helper lemma'', similar to what we did in prior proof strategies
but this time for the goal of ``squeezing'' $\alpha^n$ below $1/2$.
\begin{blemmab}
For all $\alpha < 1$ in $\mathbb{Q}_+$, $\alpha^{\text{numerator}(\alpha)} \leq 1/2$.
\end{blemmab}
\begin{proof}
Let $x/y = \alpha$.
By Lemma~\ref{lem:2-n-n-leq-1-plus-n-n}, $2y^y \leq (1+y)^y$.
By Lemma~\ref{lem:inv-monotone-div}, $(y^y)/(1+y)^y < (y^y)/(2y^y)$.
But since for all $a, b, c \in \mathbb{Q}_+, a^b/c^b = (a/c)^b$ and $a/2a = 1/2$,
it immediately follows that $(y/(1+y))^y < 1/2$.
By Lemma~\ref{lem:x-over-y-leq-x-over-1-plus-x}, $\alpha \leq \alpha/(1+\alpha)$.
Combining this with the (non-strict) monotinicity of the exponent,
clearly $\alpha^y \leq (\alpha/(1+\alpha))^y < 1/2$.
The rest follows from Lemma~\ref{lem:a-le-half-to-zero} 
(or the equivalent result, in the case of the manual proof).
\end{proof}

\subsection{Summary and Closing Thoughts}\label{subsec:closing}
In this Section we provided two rational proof strategies for Theorem~\ref{thm:a-n-0} --
the ``ceiling proof'' and the ``binomial proof'' --
both of which we implemented using ACL2s.
Since the rationals are dense in the reals, these proof strategies equally apply to the real numbers.
The ceiling proof, which is the one we used in our prior work~\cite{karn},
involved first proving an intermediary lemma about the ceiling function.
Specifically, assuming 
$0 < \alpha < 1$
and setting $k = \ceil{\alpha/(1-\alpha)}$ and $\delta = \ceil{k\alpha^k/\epsilon}$,
we proved that $\forall n \geq \delta$,
$\alpha^n \leq f_\alpha(n)$.
Since $f_\alpha(n) \leq \epsilon$, the result directly followed.
However, to prove this we first had to establish many arithmetic lemmas about the ceiling function,
so although straightforward on paper, the proof was comparatively arduous in ACL2s.

For the ``binomial proof'', we broke the problem into two steps, 
first showing that if $0 < \alpha \leq 1/2$ then $\lim_{n \to \infty} \alpha^n = 0$,
and then for any $\alpha \in (1/2, 1)$, constructing a $\delta$ such that $n > \delta \implies \alpha^n \leq 1/2$.  For the first step, we showed two different approaches, one manual and the other semi-automatic.  What made the second approach semi-automatic was that we used the termination
analysis capabilities of ACL2s to find the ``$\delta$'' for our $\epsilon/\delta$ proof automatically.
However, we still had to prove that this $\delta$ satisfied Definition~\ref{def:limit}.
For the second step, we used the binomial theorem to show that for all positive integers $p$, $2p^p \leq (p+1)^p$ and therefore, if $\alpha = p/q \leq p/(p+1)$ then $\alpha^p \leq (p/(p+1))^p \leq 1/2$.  
Overall, both versions of the binomial proof were considerably simpler (in terms of lines of code) than the ceiling proof in ACL2s, and the comparison is fair given that both strategies required importing preexisting books (refer to Table~\ref{tab:comparison}).

Next, we return to the RTO, and show how any proof of Theorem~\ref{thm:a-n-0} allows us to characterize the asymptotic behaviors of the \srtt and \rttvar.
We also discuss the implications of these results for the \rto.

\section{Analysis of RTO Calculation}\label{sec:context}
Recall from Eqn~\ref{eqn:rto} that the \rto is defined over the 
\rttvar, \srtt, and RTT sample \sample;
the \rttvar is defined over the \sample and the prior \rttvar and \srtt;
and the \srtt is defined over the \sample and prior \srtt.
Thus, going from the inside out, we begin by characterizing the \srtt;
then we use that analysis to aid our characterization of the \rttvar;
and finally we bring it all together to analyze the \rto.

All of our work will be done under the ``steady-state'' assumption defined below.
The purpose of this assumption is to define what it means for the network to exhibit bounded amounts
of oscillation.  Note however that this assumption does not limit the \emph{rate} of oscillation in the 
sample values, for example, it does not require that the samples be drawn from the image of some Lipschitz continuous function.
\begin{definition}[$c/r$ Steady State]
Let $c, r > 0$ be rationals and suppose that $\sample_i, \sample_{i+1}, \ldots, \sample_{i+n} \in [c-r, c+r]$.  Then we refer to the samples $\sample_j$ for $j = i, \ldots, i+n$ as being in a \emph{$c/r$ steady-state}.
\label{def:steady-state}
\end{definition}
For example, if $\sample_i, \sample_{i+1}, \ldots, \sample_{i+n}$ are drawn from the uniform distribution over $[12.3, 75]$ms, then they are in a $43.65/31.35$ steady-state.
Since every finite set achieves both a minimum and a maximum, all finite sequences of samples are technically speaking steady-state sequences, however, the same cannot be said for infinite sequences, such as the infinite sequence $\sample_j = 12.3 + 2j$ for $j \in \mathbb{N}$.

Without loss of generality, for the rest of this section we will assume 
samples $\sample_i, \sample_{i+1}, \ldots, \sample_{i+n}$ are in a $c/r$ steady-state.
We say WLOG because, we will consider both $n \in \mathbb{N}$ and the limit as $n \to \infty$.
We begin by analyzing the \srtt.  Recall, $\srtt_i = \sample_i$ if $i = 1$ else $(1-\alpha)\srtt_{i-1} + \alpha \sample_i$.  Since $\sample_i \in [c-r, c+r]$:
\begin{equation}
(1-\alpha)\srtt_{i-1} + \alpha(c-r) \leq \srtt_i \leq (1-\alpha)\srtt_{i-1} + \alpha(c+r)
\label{eqn:srtt-bounds-initial}
\end{equation}
Note that both bounds in Eqn.~\ref{eqn:srtt-bounds-initial} have the shape $(1-\alpha)\srtt_{i-1} + \alpha C$ for some constant $C$.  Recursing on this shape, we get the following.
\begin{lemma}
Let $C \in \mathbb{Q}_+$ and suppose that for all natural $0 \leq k \leq n$, we have
$f(k) = (1-\alpha)\srtt_{i-1} + \alpha C$.
Then the following holds for all $0 \leq k \leq n$.
\begin{equation}
\begin{aligned}
f(k) & = (1-\alpha)^{k+1} \srtt_{i-1} + \big( \sum_{j=0}^k (1-\alpha)^j \alpha \big) C \\
& = (1-\alpha)^{k+1} \srtt_{i-1} + \big( (\alpha - 1) (1-\alpha)^k + 1 \big) C
\end{aligned}
\end{equation}
\label{lem:alpha-summation}
\end{lemma}
Applying Lemma~\ref{lem:alpha-summation} to Eqn.~\ref{eqn:srtt-bounds-initial} we get the following.
\begin{theorem}Suppose 
$\sample_i, \sample_{i+1}, \ldots, \sample_{i+n}$ are in a $c/r$ steady-state.
Then $L \leq \srtt_{i+n} \leq H$ where ...
\begin{equation}
\begin{aligned}
L & = (1-\alpha)^{n+1} \srtt_{i-1} + \big( 1 - (1-\alpha)^{n+1} \big) (c-r) \text{, and } \\
H & = (1-\alpha)^{n+1} \srtt_{i-1} + \big( 1 - (1-\alpha)^{n+1} \big) (c+r) \\
\end{aligned}
\label{eqn:srtt-bounds}
\end{equation}
\end{theorem}
In \acls, our proof strategy goes as follows.  First we derive the closed form for $\sum_{j=0}^k (1-\alpha)^j \alpha$ and use it to rewrite $\srtt_{i+n}$ under the assumption that $\sample_i = \sample_{i+1} = \ldots = \sample_{i+n}$.  Then we show that in the $c/r$ steady-state scenario, the lower bound $L$ on the $\srtt_{i+n}$ is the value $\srtt_{i+n}$ would take if all the samples equaled $c-r$, and likewise the upper bound $H$ is the value $\srtt_{i+n}$ would take if all the samples equaled $c+r$.  Finally, we simplify and get the desired result.  Finally, we look at the asymptotic case.
\begin{theorem}
$\lim_{n\to\infty} L = c-r$ and $\lim_{n\to\infty} H = c+r$.
\end{theorem}
\begin{proof}
For simplicity consider: $f(n) = (1-\alpha)^{n+1} \srtt_{i-1} + \big( (\alpha - 1) (1-\alpha)^n + 1 \big) C$.
Since $0 < \alpha < 1$, we know $0 < 1-\alpha < 1$, so by Theorem~\ref{thm:a-n-0},
$(1-a)^{n+1} \to 0$, and likewise for $(1-\alpha)^n$.  This leaves only the term $C$.
Thus $\lim_{n \to \infty} f(n) = C$.  Since $L$ is just $f(n)$ with $C = c-r$ and $H$ is just $f(n)$ with $C = c+r$, the result immediately follows.
\end{proof}

Next, we consider the \rttvar calculation.  Recall that 
\(\rttvar_i = \sample_i / 2\) if $i = 1$ else 
\((1-\beta) \rttvar_{i-1} + \beta \lvert \srtt_{i-1} - \sample_i \rvert\),
where $\beta < 1$ is constant in $\mathbb{Q}_+$.
To simplify this equation, we will consider the case where
$\lvert \srtt_{i-1} - \sample_i \rvert$ is upper-bounded by some constant $\Delta$.
(Then we will show how to derive such a $\Delta$).
\begin{lemma}
Suppose 
$\sample_i, \sample_{i+1}, \ldots, \sample_{i+n}$ are in a $c/r$ steady-state,
and $\Delta > 0$ upper-bounds \(\lvert \srtt_{j-1} - \sample_j \rvert\)
for each $j = i, i+1, \ldots, i+n$.
Then:
\begin{equation}
\rttvar_{i+n} \leq (1-\beta)^{n+1} \rttvar_{i-1} + (1-(1-\beta)^{n+1})\Delta
\label{eqn:upper-bound-rttvar}
\end{equation}
\end{lemma}
Now we want to derive a bound on \(\lvert \srtt_{j-1} - \sample_j \rvert\).
Note that this bound is at most:
\begin{equation}
\max\{ \lvert L - (c+r) \rvert, \lvert H - (c-r) \rvert \}
\end{equation}
If the larger of the two options is $\lvert H - (c-r) \rvert$ we derive the following.
\begin{equation}
\begin{aligned}
\Delta & = \lvert (1-\alpha)^{j+1} \srtt_{i-1} +
(1-(1-\alpha)^{j+1}) (c+r) -
(c-r) \rvert \\
& = \lvert (1-\alpha)^{j+1} \srtt_{i-1} +
(c+r) - (1-\alpha)^{j+1}(c+r) - c + r \rvert \\
& = \lvert (1-\alpha)^{j+1} \srtt_{i-1} + 2r - (1-\alpha)^{j+1}(c+r) \rvert
\end{aligned}
\end{equation}
By Theorem~\ref{thm:a-n-0}, clearly $\lim_{j \to \infty} \Delta = 2r$.  
On the other hand, if the larger option is $\lvert L - (c+r) \rvert$, we derive
$\Delta = \lvert (1-\alpha)^{j+1} \srtt_{i-1} - 2r + (1-\alpha)^{j+1}(c-r) \rvert$ which asymptotes at $\lvert  -2r \rvert = 2r$.
So either way, as $j$ grows $\to \infty$, $\Delta$ converges to $2r$.
Now suppose $\Delta = 2r$.  Then the upper bound on $\srtt_{i+n}$ from Eqn~\ref{eqn:upper-bound-rttvar}
is $(1-\beta)^{n+1} \rttvar_{i-1} + (1-(1-\beta)^{n+1})2r$.
Which gives us the following.
\begin{theorem}
Suppose 
$\sample_i, \sample_{i+1}, \ldots, \sample_{i+n}$ are in a $c/r$ steady-state.
Then there exists an upper bound on $\rttvar_{i+n}$ which, as $n \to \infty$, converges to $2r$.
\end{theorem}
\begin{proof}Follows from Theorem~\ref{thm:a-n-0} because $\lim_{n \to \infty} (1-\beta)^{n+1} \rttvar_{i-1} = \lim_{n \to \infty} (1-\beta)^{n+1} 2r = 0$.
\end{proof}

Finally, we turn our attention to the \rto calculation.  Recall, \(\rto_i = \srtt_i + \max(G, 4\cdot \rttvar_i)\) for some constant $G$.  On the one hand, if the $\rttvar$ is consistently very small (less than 1/4 of G) then clearly the $\rto$ is bounded by $[L + G, H + G]$.  In this case, if $G > 2r$, we are assured that timeouts will never happen.  But what if $G$ is small relative to the radius of the steady-state interval?  If $G < 2r$ and the $\rttvar$ can achieve a value $\leq G$ then a timeout can occur.  And in fact, we can easily construct a scenario where exactly this happens infinitely many times.

For the pathological scenario, suppose that every 100$^{\text{th}}$ sample equals 75, while all the rest equal 60.  Clearly the samples are in a $67.5/7.5$ steady-state.  At the spikes (where $\sample_{i + 100n} = 75$), $\srtt \approx 61.88, \rttvar \approx 3.75$, and $\rto \approx 61$.  Since $61 < 75$, a timeout occurs.  There are infinitely many ``spikes'' where timeouts occur.  However, when we simulate a scenario where the samples are uniformly random over $[c-r, c+r]$, little to no timeouts occur.  Both scenarios are illustrated below in Fig.~\ref{fig:scenarios}.\footnote{Adapted from Fig.~3 of~\cite{karn}, first published in volume 14067, page 56, 2023, by Springer Nature.}

\begin{figure}[h]
\centering
\includegraphics[width=\textwidth]{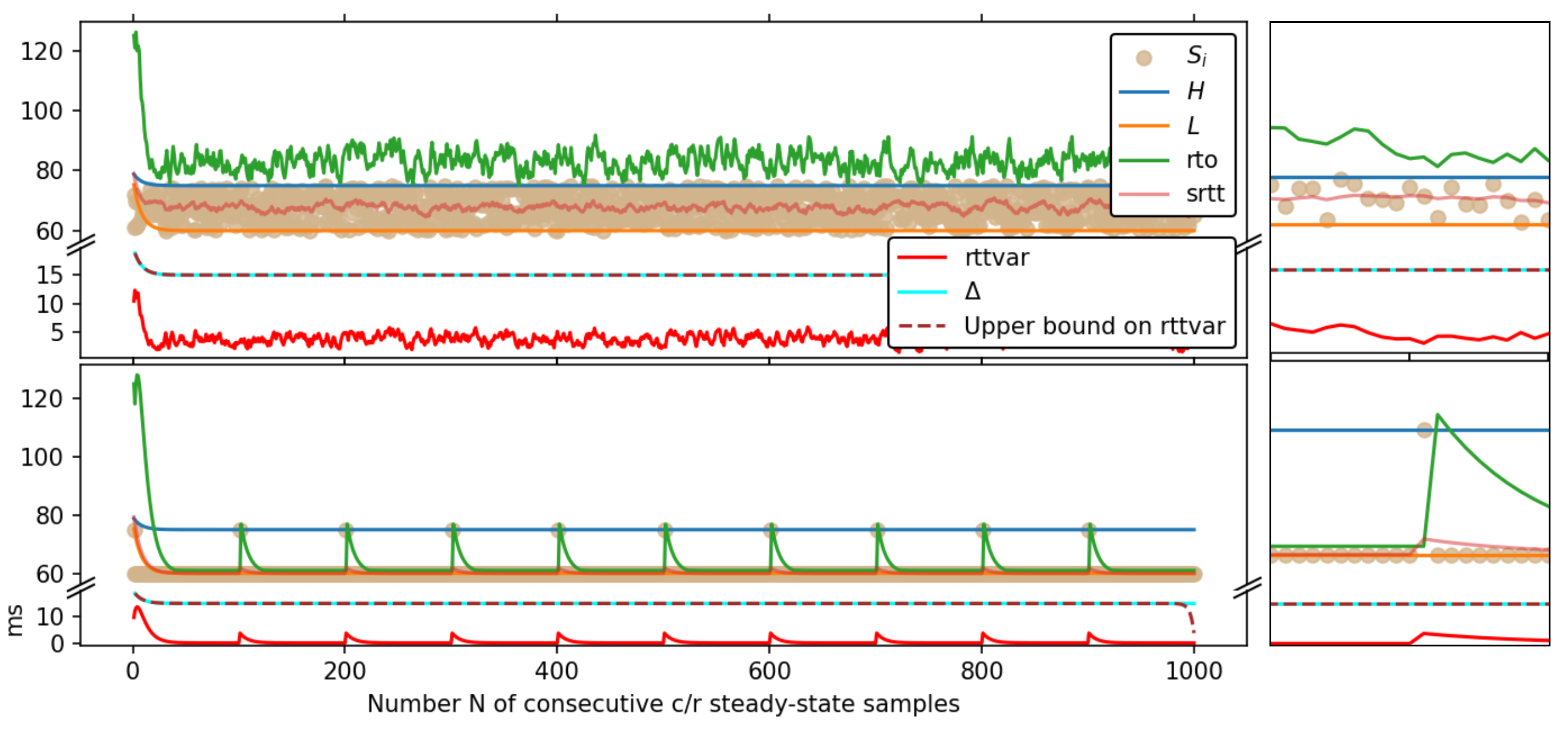}
\caption{On the left are two $67.5/7.5$ steady-state scenarios.  On top the samples are drawn from the uniform distribution over the bounds, and timeouts rarely, if ever, occur.  In the bottom (pathological) scenario, every 100$^{\text{th}}$ sample equals $c+r = 75$ while the rest equal $c-r = 60$, and at each ``spike'', a timeout occurs.  There are infinitely many spikes, and one is shown on the right ($n=[350,450]$).}
\label{fig:scenarios}
\end{figure}

The RTO calculation is specified in RFC6298~\cite{rfc6298} which says the constant~$G$ should be set to the ``clock granularity'' of the sender, in seconds.  In the interest of avoiding excessive timeouts, a protocol implementer might want to consider adding the additional criteria that $G$ should exceed 1/2 the diameter of the maximally large sample interval that they consider to be ``stable''.  Depending on the nature and context of the protocol, this could be either a fixed or dynamic value.  However, caution should be taken on the other hand to ensure $G$ is not too large, since timeouts \emph{should} occur when there really is congestion on the network, in order to avoid congestion collapse.  Also, a dynamic value of $G$ could exacerbate the risk of choosing too large of a timeout value, and might even be vulnerable to targeted manipulation.

As we alluded when introducing Def~\ref{def:steady-state}, one interesting direction for future research is to investigate minimal, sufficient analytic conditions to ensure timeouts do not occur.  For example, it might be sufficient to require the samples be drawn from the image of a Lipschitz continuous function, and then to require 
some relationship between the Lipschitz bound, the choice of $\beta$, and the radius~$r$.
Although this is purely speculative, it is certainly the case that the derivative of the samples
must be taken into account when analyzing the magnitude of the \rttvar, providing 
multiple interesting directions for future research.

\section{Discussion}\label{sec:discussion}
In this work we considered multiple approaches to proving $\forall \alpha \in [0, 1) \,::\, \lim_{n \to \infty} \alpha^n = 0$.  
Our first was in ACL2(r) and 
resembled the 
obvious pen-and-paper proof, but could not be imported into an ACL2 environment.
Our second and third were in ACL2s and required proving some arithmetic lemmas.
Of these, the semi-automated version of the third is the most stylistically aligned with ACL2s because it takes advantage of automated termination analysis.
The ACL2(r) proof has the lowest character count of all the proofs,
and certifies the most quickly,
but imports the most books and is not portable to ACL2.
Considering books and portability, 
the semi-automated binomial proof 
    is probably the best (see Tab.~\ref{tab:comparison}).

\begin{table}[]
\centering
\begin{tabular}{lllllll}
Proof         & LoC & Chars  & Props/Thms & Functions & Books & Cert Time (s) \\\hline
Real          & 161 & 4,224  & 17         & 1         & 5     & 0.58         \\
Ceiling       & 408 & 16,103 & 20         & 3         & 0     & 64.17         \\
Binomial (M)  & 154 & 5,652  & 22         & 1         & 2     & 2.54          \\
Binomial (SA) & 122 & 5,402  & 22         & 2         & 1     & 3.84         
\end{tabular}
\caption{Proof comparison.  (M) refers to ``manual'' while (SA) refers to ``semi-automatic''.  Lines of code and character count are computed without comments or empty lines, however, the proofs are not styled identically.  Props/Thms counts instances of \codify{property} and \codify{defthm}, while Functions counts \codify{definec}s, \codify{definecd}s, and \codify{defun}s.  Certification time is measured on a 16GB M1 Macbook Air.}
\label{tab:comparison}
\end{table}

While working on the third approach we encountered an inconvenience in ACL2s:
in ACL2s, 
\\\codify{definec}s support \codify{function-} and \codify{body-contract-hints}, but
\codify{property} definitions do not.
We found two ways to address this.  The first was to redefine the problematic property as a decision procedure in a \codify{definec} with appropriate
hints, and then write a new property saying that on all inputs, the procedure returns \emph{true}.
The second was simply to set \codify{:check-contracts? nil}.
To ameliorate this issue, we plan to add hints to \codify{property} definitions in a future update to ACL2s.

While writing the ACL2s proofs, we took advantage of ACL2s features not available in ACL2 or ACL2(r).  Most notably, we used termination analysis in the semi-automated binomial proof.  There are also smaller ways that ACL2s was easier: type annotations improved the readability of our proofs, and contract-checking gave us some lemmas for free, including type-checking in the \codify{xargs} to our \codify{defun-sk}.  In contrast, we had to manually prove contracts for $d_\alpha$ in our ACL2(r) proof, and it is harder to read than either ACL2s proof (with lengthier antecedents on \codify{defthm}s) because it lacks type annotations.

\section{Conclusion}\label{sec:conclusion}
Recently, the Internet Engineering Task Force created a Usable Formal Methods Research Group,
of which we are members,
to integrate formal methods into the RFC drafting process.
In many protocols, \emph{performance} matters:
we want to know how quickly the protocol achieves a desired outcome under load.
Usually performance is studied using simulations or measurements, 
which can give a sense of how protocols behave ``in the wild''.  
But what about how protocols behave in the worst case?  
What if the worst case never happens in 
the measured environments or simulations?  For this, we need a 
way to \emph{prove} performance bounds, namely, formal methods.
Meanwhile, real analysis provides a convenient framework with which to ask and answer
questions about performance bounds in the long run.  
In our prior work~\cite{karn}, we used formal methods with real analysis
 to prove useful bounds on the internal variables of the RTO calculation.
But we also ran into hurdles.
We could not use the most obvious proof, which requires real numbers,
because ACL2s only supports rationals.
When we came up with an alternative approach (the ``ceiling proof'')
it required us to convince ACL2s of numerous arithmetic lemmas.
Only post-publication did we find a simpler solution,
based on the binomial theorem.

As relatively novice users of ACL2s\footnote{Excluding the second author.}, our work leads us to identify three areas where we feel the ACL2 ecosystem could be improved to support work such as ours.
First, ACL2 (and ACl2s in particular) 
could benefit from a richer,
more comprehensively documented,
and more easily searchable library of purely mathematical theorems, relating to the ceiling, floor, exponent, and logarithm, as well as metric spaces and limits.  
Searching for proofs is difficult enough, and ACL2 does not come with any kind of semantic proof search tool.
And often, even when the desired theorems exist in the ACL2 books, they are unmentioned in the documentation.
For example, the documentation on ``arithmetic'' does not mention the RTL books, 
and neither does the documentation on ``math''.  Moreover, the rewrite rules from different libraries may conflict,
so even if you find the desired theorems, importing them into a singular environment
may be non-trivial.
It would also be useful to have more mathematics formalized in 
ACL2s, so as to avoid additional proof obligations (for function contracts and termination) when using imported books.
Second, ACL2(r) could benefit from the addition of the generic exponent and logarithm.  
This could be done using the translational Lemma given in $\mathsection$\ref{sec:real}.
Third, and most importantly, though ACL2(r) and ACL2 have incompatible theories, it is nevertheless true that certain kinds of theorems over the reals should hold over the rationals, because the rationals are dense in the reals.  It would useful to have a kind of ``bridge'' between ACL2(r) and ACL2, by which the user could justify that a given theorem, if true over the reals, must also hold over the rationals; prove the theorem in ACL2(r); and then import the theorem, using its ``justification'', into ACL2.
Hopefully our experience provides insight for future work in both protocol analysis and extending the ACL2 ecosystem.

\vspace{1em}

\noindent \textbf{Acknowledgments}. The first author would like to thank Ruben Gamboa for providing technical support in ACL2(r) and suggesting Lemma~\ref{lem:e-n-a-to-a-n}, and Ankit Kumar for providing technical support in ACL2s.

\nocite{*}
\bibliographystyle{eptcs}
\bibliography{main}
\end{document}